\documentclass[fleqn]{llncs}
\usepackage{version}
\usepackage{pgabr}

\title{A Short Introduction to Program Algebra \\ with Instructions for
       Boolean Registers}
\author{J.A. Bergstra \and C.A. Middelburg}
\institute{Informatics Institute, Faculty of Science,
           University of Amsterdam, \\
           Science Park~904, 1098~XH Amsterdam, the Netherlands \\
           \email{J.A.Bergstra@uva.nl,C.A.Middelburg@uva.nl}}

\begin{document}

\maketitle

\begin{abstract}
A parameterized algebraic theory of instruction sequences, objects that 
represent the behaviours produced by instruction sequences under 
execution, and objects that represent the behaviours exhibited by the
components of the execution environment of instruction sequences is the 
basis of a line of research in which issues relating to a wide variety 
of subjects from computer science have been rigorously investigated 
thinking in terms of instruction sequences.
In various papers that belong to this line of research, use is made of 
an instantiation of this theory in which the basic instructions are
instructions to read out and alter the content of Boolean registers and 
the components of the execution environment are Boolean registers.
In this paper, we give a simplified presentation of the most general 
such instantiated theory. 
\begin{classcode}
F.1.1, F.3.2, F.4.1.
\end{classcode}
\end{abstract}

\section{Introduction}
\label{sect-intro}

We are carrying out a line of research in which issues relating to a 
wide variety of subjects from computer science are rigorously 
investigated thinking in terms of instruction sequences (see
e.g.~\cite{SiteIS}).
The groundwork for this line of research is the combination of an 
algebraic theory of single-pass instruction sequences, called program 
algebra, and an algebraic theory of mathematical objects that represent 
the behaviours produced by instruction sequences under execution, called 
basic thread algebra, extended to deal with the interaction between 
instruction sequences under execution and components of their execution 
environment (see~e.g.~\cite{BM12b}).
This groundwork is parameterized by a set of basic instructions and a 
set of objects that represent the behaviours exhibited by the components 
of the execution environment.

In various papers that have resulted from this line of research, use is 
made of an instantiation of this theory in which certain instructions to 
read out and alter the content of Boolean registers are taken as basic 
instructions and Boolean registers are taken as the components of the 
execution environment 
\mbox{(see \cite{BM13b,BM13a,BM14a,BM14e,BM13c,BM18a})}.
In the current paper, we give a simplified presentation of the 
instantiation in which all possible instructions to read out and alter 
the content of Boolean registers are taken as basic instructions.

In the papers referred to above, the rationale for taking certain 
instructions to read out and alter the content of Boolean registers as
basic instructions is that the instructions concerned are sufficient to 
compute each function on bit strings of any fixed length by a finite 
instruction sequence.
However, shorter instruction sequences may be possible if certain 
additional instructions to read out and alter the content of Boolean 
registers are taken as basic instructions (see~\cite{BM15a}).
That is why we opted for the most general instantiation.

Both program algebra and basic thread algebra were first presented 
in~\cite{BL02a}.%
\footnote
{In that paper and the first subsequent papers, basic thread algebra 
 was introduced under the name basic polarized process algebra.}
An extension of basic thread algebra to deal with the interaction 
between instruction sequences under execution and components of their 
execution environment, called services, was presented for the first time
in~\cite{BP02a}. 
A substantial re-design of this extension was first presented 
in~\cite{BM09k}.
The presentation of both extensions is rather involved because they are
parameterized and owing to this cover a generic set of basic 
instructions and a generic set of services.
In the current paper, a much less involved presentation is obtained by 
covering only the case where the basic instructions are instructions to 
read out and alter the content of Boolean registers and the services are 
Boolean registers.

This paper is organized as follows.
First, we introduce program algebra (Section~\ref{sect-PGA}) and basic 
thread algebra (Section~\ref{sect-BTA}) and extend their combination to 
make precise which behaviours are produced by instruction sequences 
under execution (Section~\ref{sect-TE-BC}).
Next, we present the instantiation of the resulting theory in which all 
possible instructions to read out and alter Boolean registers are taken 
as basic instructions (Section~\ref{sect-PGAbr}), introduce an algebraic 
theory of Boolean register families (Section~\ref{sect-BRFA}), and 
extend the combination of the theories presented in the two preceding 
sections to deal with the interaction between instruction sequences 
under execution and Boolean registers (Section~\ref{sect-TSI}).
Then, we formalize in the setting of the resulting theory what it means 
that a given instruction sequence computes a given partial function from 
$\Bool^n$ to $\Bool^m$ ($n,m \in \Nat$) 
(Section~\ref{sect-comput-bool-fnc}) and give a survey of uses for the 
resulting theory (Section~\ref{sect-work-done}). 
Finally, we make some concluding remarks (Section~\ref{sect-concl}).

The following should be mentioned in advance.
The set $\Bool$ is a set with two elements whose intended 
interpretations are the truth values \emph{false} and \emph{true}.
As is common practice, we represent the elements of $\Bool$ by the bits 
$0$ and $1$.
In line with generally accepted conventions, we use terminology based on
identification of the elements of $\Bool$ with their representation 
where appropriate.
For example, the elements of $\Bool^n$ are loosely called bit strings of 
length $n$.

In this paper, some familiarity with algebraic specification is assumed.
The relevant notions are explained in handbook chapters and books on 
algebraic specification, e.g.~\cite{EM85a,ST99a,ST12a,Wir90a}.

This paper is to a large extent a compilation of material from several 
earlier publications.
Various examples, various explanatory remarks, and the axioms from 
Section~\ref{sect-TSI} do not occur in earlier publications. 

\section{Program Algebra}
\label{sect-PGA}

In this section, we present \PGA\ (ProGram Algebra).
The starting-point of \PGA\ is the perception of a program as a 
single-pass instruction sequence, i.e.\ a possibly infinite sequence of 
instructions of which each instruction is executed at most once and can 
be dropped after it has been executed or jumped over.
The concepts underlying the primitives of program algebra are common in
programming, but the particular form of the primitives is not common.
The predominant concern in the design of \PGA\ has been to achieve 
simple syntax and semantics, while maintaining the expressive power of 
arbitrary finite control.

It is assumed that a fixed but arbitrary set $\BInstr$ of 
\emph{basic instructions} has been given.
$\BInstr$ is the basis for the set of instructions that may occur in 
the instruction sequences considered in \PGA.
The intuition is that the execution of a basic instruction may modify a
state and must produce the Boolean value $\False$ or $\True$ as reply at 
its completion.
The actual reply may be state-dependent.

In applications of \PGA, the instructions taken as basic instructions 
vary, in effect, from instructions relating to unbounded counters, 
unbounded stacks or Turing tapes through instructions relating to 
Boolean registers or natural number registers to machine language 
instructions of actual computers.   

The set of instructions of which the instruction sequences considered 
in \PGA\ are composed is the set that consists of the following 
elements:
\begin{itemize}
\item
for each $a \in \BInstr$, a \emph{plain basic instruction} $a$;
\item
for each $a \in \BInstr$, a \emph{positive test instruction} $\ptst{a}$;
\item
for each $a \in \BInstr$, a \emph{negative test instruction} $\ntst{a}$;
\item
for each $l \in \Nat$, a \emph{forward jump instruction} $\fjmp{l}$;
\item
a \emph{termination instruction} $\halt$.
\end{itemize}
We write $\PInstr$ for this set.
The elements from this set are called \emph{primitive instructions}.

Primitive instructions are the elements of the instruction sequences 
considered in \PGA.
On execution of such an instruction sequence, these primitive 
instructions have the following effects:
\begin{itemize}
\item
the effect of a positive test instruction $\ptst{a}$ is that basic
instruction $a$ is executed and execution proceeds with the next
primitive instruction if $\True$ is produced and otherwise the next
primitive instruction is skipped and execution proceeds with the
primitive instruction following the skipped one --- if there is no
primitive instruction to proceed with,
inaction occurs;
\item
the effect of a negative test instruction $\ntst{a}$ is the same as
the effect of $\ptst{a}$, but with the role of the value produced
reversed;
\item
the effect of a plain basic instruction $a$ is the same as the effect
of $\ptst{a}$, but execution always proceeds as if $\True$ 
is produced;
\item
the effect of a forward jump instruction $\fjmp{l}$ is that execution
proceeds with the $l$th next primitive instruction --- if $l$ equals $0$ 
or there is no primitive instruction to proceed with, inaction occurs;
\item
the effect of the termination instruction $\halt$ is that execution 
terminates.
\end{itemize}
Inaction occurs if no more basic instructions are executed, but 
execution does not terminate.

A plain basic instruction $a$ is generally used in the case of a basic 
instruction $a$ that modifies a state and a positive test instruction 
$\ptst{a}$ or a negative test instruction $\ntst{a}$ is generally used 
in the case of a basic instruction $a$ that does not modify a state.
However, there are no rules prescribing such use.

\PGA\ has one sort: the sort $\InSeq$ of \emph{instruction sequences}. 
We make this sort explicit to anticipate the need for many-sortedness
later on.
To build terms of sort $\InSeq$, \PGA\ has the following constants and 
operators:
\begin{itemize}
\item
for each $u \in \PInstr$, 
the \emph{instruction} constant $\const{u}{\InSeq}$\,;
\item
the binary \emph{concatenation} operator 
$\funct{\ph \conc \ph}{\InSeq \x \InSeq}{\InSeq}$\,;
\item
the unary \emph{repetition} operator 
$\funct{\ph\rep}{\InSeq}{\InSeq}$\,.
\end{itemize}
Terms of sort $\InSeq$ are built as usual in the one-sorted case.
We assume that there are infinitely many variables of sort $\InSeq$, 
including $X,Y,Z$.
We use infix notation for concatenation and postfix notation for
repetition.
Taking these notational conventions into account, the syntax of closed 
\PGA\ terms (of sort $\InSeq$) can be defined in Backus-Naur style as 
follows:
\begin{ldispl}
\nm{CT}_\InSeq \Sis 
a \Sor \ptst{a} \Sor \ntst{a} \Sor \fjmp{l} \Sor \halt \Sor 
(\nm{CT}_\InSeq \conc \nm{CT}_\InSeq) \Sor ({\nm{CT}_\InSeq} \rep)\;,
\end{ldispl}%
where $a \in \BInstr$ and $l \in \Nat$.%
\footnote
{We use $\nm{CT}_\mathbf{S}$, where $\mathbf{S}$ is a sort, as 
 nonterminal standing for closed terms of sort~$\mathbf{S}$.}

Throughout the paper, we generally omit grouping parentheses if they can 
be unambiguously added or they are unnecessary because it is axiomatized 
that the operator concerned stands for an associative operation. 

A \PGA\ term in which the repetition operator does not occur is called 
a \emph{repetition-free} \PGA\ term.
A \PGA\ term that is not repetition-free is said to be a \PGA\ term that
\emph{has a repeating part}. 

One way of thinking about closed \PGA\ terms is that they represent 
non-empty, possibly infinite sequences of primitive instructions with 
finitely many distinct suffixes.
The instruction sequence represented by a closed term of the form
$t \conc t'$ is the instruction sequence represented by $t$
concatenated with the instruction sequence represented by $t'$.%
\footnote
{The concatenation of an infinite sequence with a finite or infinite 
sequence yields the former sequence.}
The instruction sequence represented by a closed term of the form 
$t\rep$ is the instruction sequence represented by $t$ concatenated 
infinitely many times with itself.
A closed \PGA\ term represents a finite instruction sequence if and 
only if it is a closed repetition-free \PGA\ term.
 
A simple example of a closed \PGA\ term is
\begin{ldispl}
(\ntst{a} \conc (\fjmp{3} \conc (b \conc \halt)))\rep\;.
\end{ldispl}%
On execution of the infinite instruction sequence denoted by this term, 
first the basic instruction $a$ is executed repeatedly until its 
execution produces the reply $\True$, next the basic instruction $b$ is 
executed, and after that execution terminates.
Because $(X \conc Y) \conc Z = X \conc (Y \conc Z)$ is an axiom of \PGA\ 
(see below), we could have written
$(\ntst{a} \conc \fjmp{3} \conc b \conc \halt) \rep$ instead of 
$(\ntst{a} \conc (\fjmp{3} \conc (b \conc \halt)))\rep$ above.

The axioms of \PGA\ are given in Table~\ref{axioms-PGA}.%
\begin{table}[!t]
\caption{Axioms of \PGA} 
\label{axioms-PGA}
\begin{eqntbl}
\begin{axcol}
(X \conc Y) \conc Z = X \conc (Y \conc Z)             & \axiom{PGA1}  \\
(X^n)\rep = X\rep                                     & \axiom{PGA2}  \\
X\rep \conc Y = X\rep                                 & \axiom{PGA3}  \\
(X \conc Y)\rep = X \conc (Y \conc X)\rep             & \axiom{PGA4} 
\eqnsep
\fjmp{k{+}1} \conc u_1 \conc \ldots \conc u_k \conc \fjmp{0} =
\fjmp{0} \conc u_1 \conc \ldots \conc u_k \conc \fjmp{0} 
                                                      & \axiom{PGA5}  \\
\fjmp{k{+}1} \conc u_1 \conc \ldots \conc u_k \conc \fjmp{l} =
\fjmp{l{+}k{+}1} \conc u_1 \conc \ldots \conc u_k \conc \fjmp{l}
                                                      & \axiom{PGA6}  \\
(\fjmp{l{+}k{+}1} \conc u_1 \conc \ldots \conc u_k)\rep =
(\fjmp{l} \conc u_1 \conc \ldots \conc u_k)\rep       & \axiom{PGA7}  \\
\fjmp{l{+}k{+}k'{+}2} \conc u_1 \conc \ldots \conc u_k \conc
(v_1 \conc \ldots \conc v_{k'{+}1})\rep = {} \\ \phantom{{}{+}k'}
\fjmp{l{+}k{+}1} \conc u_1 \conc \ldots \conc u_k \conc
(v_1 \conc \ldots \conc v_{k'{+}1})\rep               & \axiom{PGA8} 
\end{axcol}
\end{eqntbl}
\end{table}
In this table, 
$u$, $u_1,\ldots,u_k$ and $v_1,\ldots,v_{k'+1}$ stand for arbitrary 
primitive instructions from $\PInstr$, 
$k$, $k'$, and $l$ stand for arbitrary natural numbers from $\Nat$, and
$n$ stands for an arbitrary natural number from $\Natpos$.%
\footnote
{We write $\Natpos$ for the set $\set{n \in \Nat \where n \geq 1}$ of
positive natural numbers.}
For each $n \in \Natpos$, the term $t^n$, where $t$ is a \PGA\ term, 
is defined by induction on $n$ as follows: $t^1 = t$, and 
$t^{n+1} = t \conc t^n$.

Some simple examples of equations derivable from the axioms of \PGA\ are
\begin{ldispl}
(a \conc b)\rep \conc c = a \conc (b \conc a)\rep\;,
\\
\ptst{a} \conc (b \conc (\ntst{c} \conc \fjmp{2} \conc \halt)\rep)\rep
=
\ptst{a} \conc b \conc (\ntst{c} \conc \fjmp{2} \conc \halt)\rep\;.
\end{ldispl}%

Closed  \PGA\ terms $t$ and $t'$ represent the same instruction sequence 
iff $t = t'$ is derivable from PGA1--PGA4.
In this case, we say that the represented instruction sequences are 
\emph{instruction sequence congruent}.
We write \PGAisc\ for the algebraic theory whose sorts, constants and
operators are those of \PGA, but whose axioms are PGA1--PGA4.

The informal explanation of closed \PGA\ terms as sequences of primitive 
instructions given above can be looked upon as a sketch of the intended 
model of the axioms of \PGAisc.
This model, which is described in detail in, for example, \cite{BM12b}, 
is an initial model of the axioms of \PGAisc.

The \emph{unfolding equation} $X\rep = X \conc X\rep$ is derivable from
the axioms of \PGAisc\ by first taking the instance of PGA2 in which 
$n = 2$, then applying PGA4, and finally applying the instance of PGA2 
in which $n = 2$ again.

A closed \PGA\ term is in \emph{first canonical form} if it is of the 
form $t$ or $t \conc {t'}\rep$, where $t$ and $t'$ are closed 
repetition-free \PGA\ terms.
The following proposition, proved in~\cite{BM12b}, relates \PGAisc\ and 
first canonical forms.
\begin{proposition}
\label{prop-1CF}
For all closed \PGA\ terms $t$, there exists a closed \PGA\ term $t'$ 
that is in first canonical form such that $t = t'$ is derivable from
the axioms of \PGAisc.
\end{proposition}
The examples given above of equations derivable from the axioms of \PGA\ 
are derivable from the axioms of \PGAisc\ only.
Their left-hand sides are not in first canonical form and their 
right-hand sides are in first canonical form.
Simple examples of equations derivable from the axioms of \PGA\ and not 
derivable from the axioms of \PGAisc\ are
\begin{ldispl}
\ntst{a} \conc \fjmp{2} \conc (\ptst{b} \conc \fjmp{2})\rep
=
\ntst{a} \conc \fjmp{0} \conc (\ptst{b} \conc \fjmp{0})\rep\;,
\\
\ptst{a} \conc \fjmp{6} \conc b \conc (\ntst{c} \conc \fjmp{9})\rep
=
\ptst{a} \conc \fjmp{2} \conc b \conc (\ntst{c} \conc \fjmp{1})\rep\;.
\end{ldispl}%

Closed  \PGA\ terms $t$ and $t'$ represent the same instruction sequence 
after changing all chained jumps into single jumps and making all jumps 
as short as possible iff $t = t'$ is derivable from PGA1--PGA8. 
In this case, we say that the represented instruction sequences are 
\emph{structurally congruent}.

A closed \PGA\ term $t$ \emph{has chained jumps} if there exists a 
closed \PGA\ term $t'$ such that $t = t'$ is derivable from the axioms
of \PGAisc\ and $t'$ contains a subterm of the form 
$\fjmp{n{+}1} \conc u_1 \conc \ldots \conc u_n \conc \fjmp{l}$.
A closed \PGA\ term $t$ of the form
$u_1 \conc \ldots \conc u_m \conc (v_1 \conc \ldots \conc v_k)\rep$
\emph{has shortest possible jumps} if:
(i)~for each $i \in [1,m]$ for which $u_i$ is of the form $\fjmp{l}$,
$l \leq k + m - i$;
(ii)~for each $j \in [1,k]$ for which $v_j$ is of the form $\fjmp{l}$,
$l \leq k - 1$.
A closed \PGA\ term is in \emph{second canonical form} if it is in first 
canonical form, does not have chained jumps, and has shortest possible 
jumps if it has a repeating part.
The following proposition, proved in~\cite{BM12b}, relates \PGA\ and 
second canonical forms.
\begin{proposition}
\label{prop-2CF}
For all closed \PGA\ terms $t$, there exists a closed \PGA\ term $t'$ 
that is in second canonical form such that $t = t'$ is derivable from 
the axioms of \PGA.
\end{proposition}
The examples given above of equations derivable from the axioms of \PGA\ 
and not derivable from the axioms of \PGAisc\ have left-hand sides that
are not in second canonical form and right-hand sides that are in second 
canonical form.

Henceforth, the instruction sequences of the kind considered in \PGA\ 
are called \PGA\ instruction sequences.

In Section~\ref{sect-TSI}, we will use the notation 
$\Conc{i = 1}{n} t_i$.
For each $i \in \Natpos$, let $t_i$ be \PGA\ terms.
Then, for each $n \in \Natpos$, the term $\Conc{i = 1}{n} t_i$ is 
defined by induction on $n$ as follows: $\Conc{i = 1}{1} t_i = t_1$ 
and $\Conc{i = 1}{n+1} t_i = \Conc{i = 1}{n} t_i \conc t_{n+1}$.

\section{Basic Thread Algebra for Finite and Infinite Threads}
\label{sect-BTA}

In this section, we present \BTA\ (Basic Thread Algebra) and an 
extension of \BTA\ that reflects the idea that infinite threads are 
identical if their approximations up to any finite depth are identical.

\BTA\ is concerned with mathematical objects that model in a direct 
way the behaviours produced by \PGA\ instruction sequences under 
execution.
The objects in question are called threads.
A thread models a behaviour that consists of performing basic actions in 
a sequential fashion.
Upon performing a basic action, a reply from an execution environment
determines how the behaviour proceeds subsequently.
The possible replies are the Boolean values $\False$ and $\True$.

The basic instructions from $\BInstr$ are taken as basic actions.
Besides, $\Tau$ is taken as a special basic action.
It is assumed that $\Tau \notin \BAct$.
We write $\BActTau$ for $\BAct \union \set{\Tau}$.

\BTA\ has one sort: the sort $\Thr$ of \emph{threads}. 
We make this sort explicit to anticipate the need for many-sortedness
later on.
To build terms of sort $\Thr$, \BTA\ has the following constants and 
operators:
\begin{itemize}
\item
the \emph{inaction} constant $\const{\DeadEnd}{\Thr}$;
\item
the \emph{termination} constant $\const{\Stop}{\Thr}$;
\item
for each $\alpha \in \BActTau$, the binary 
\emph{postconditional composition} operator 
$\funct{\pcc{\ph}{\alpha}{\ph}}{\Thr \x \Thr}{\Thr}$.
\end{itemize}
Terms of sort $\Thr$ are built as usual in the one-sorted case. 
We assume that there are infinitely many variables of sort $\Thr$, 
including $x,y,z$.
We use infix notation for postconditional composition.
Taking this notational convention into account, the syntax of closed 
\BTA\ terms (of sort $\Thr$) can be defined in Backus-Naur style as 
follows:
\begin{ldispl}
\nm{CT}_\Thr \Sis 
\DeadEnd \Sor \Stop \Sor 
(\pcc{\nm{CT}_\Thr}{\alpha}{\nm{CT}_\Thr})\;, 
\end{ldispl}%
where $\alpha \in \BActTau$.
We introduce \emph{basic action prefixing} as an abbreviation: 
$\alpha \bapf t$, where $\alpha \in \BActTau$ and $t$ is a \BTA\ term, 
abbreviates $\pcc{t}{\alpha}{t}$.
We treat an expression of the form $\alpha \bapf t$ and the \BTA\ term 
that it abbreviates as syntactically the same.

Closed \BTA\ terms are considered to represent threads.
The thread represented by a closed term of the form 
$\pcc{t}{\alpha}{t'}$ models the behaviour that first performs $\alpha$, 
and then proceeds as the behaviour modeled by the thread represented by 
$t$ if the reply from the execution environment is $\True$ and proceeds 
as the behaviour modeled by the thread represented by $t'$ if the reply 
from the execution environment is $\False$. 
Performing $\Tau$, which is considered performing an internal action,
always leads to the reply $\True$.
The thread represented by $\Stop$ models the behaviour that does nothing 
else but terminate and the thread represented by $\DeadEnd$ models the 
behaviour that is inactive, i.e.\ it performs no more basic actions and 
it does not terminate. 

A simple example of a closed \BTA\ term is
\begin{ldispl}
\pcc{(b \bapf \Stop)}{a}{\DeadEnd}\;.
\end{ldispl}%
This term denotes the thread that first performs basic action $a$, if 
the reply from the execution environment on performing $a$ is $\True$, 
it next performs the basic action $b$ and then terminates, and if the 
reply from the execution environment on performing $a$ is $\False$, it 
next becomes inactive.

\BTA\ has only one axiom.
This axiom is given in Table~\ref{axioms-BTA}.
\begin{table}[!t]
\caption{Axioms of \BTA} 
\label{axioms-BTA}
\begin{eqntbl}
\begin{axcol}
\pcc{x}{\Tau}{y} = \pcc{x}{\Tau}{x}                      & \axiom{T1}
\end{axcol}
\end{eqntbl}
\end{table}
Using the abbreviation introduced above, it can also be written as
follows: $\pcc{x}{\Tau}{y} = \Tau \bapf x$.

Each closed \BTA\ term represents a finite thread, i.e.\ a thread with 
a finite upper bound to the number of basic actions that it can perform.
Infinite threads, i.e.\ threads without a finite upper bound to the
number of basic actions that it can perform, can be defined by means of 
a set of recursion equations (see e.g.~\cite{BM09k}).

A simple example of a set of recursion equations that consists of a 
single equation is
\begin{ldispl}
x = \pcc{(b \bapf \Stop)}{a}{x}\;.
\end{ldispl}%
Its solution is the thread that first repeatedly performs basic action 
$a$ until the reply from the execution environment on performing $a$ is 
$\True$, next performs the basic action $b$ and then terminates.

A regular thread is a finite or infinite thread that can be defined by 
means of a finite set of recursion equations.
The behaviours produced by \PGA\ instruction sequences under execution 
are exactly the behaviours modeled by regular threads.

Two infinite threads are considered identical if their approximations up 
to any finite depth are identical.
The approximation up to depth $n$ of a thread models the behaviour that 
differs from the behaviour modeled by the thread in that it will become
inactive after it has performed $n$ actions unless it would terminate at
this point.
AIP (Approximation Induction Principle) is a conditional equation that
formalizes the above-mentioned view on infinite threads.
In AIP, the approximation up to depth $n$ is phrased in terms of the
unary \emph{projection} operator $\funct{\proj{n}}{\Thr}{\Thr}$.

The axioms for the projection operators and AIP are given in
Table~\ref{axioms-BTAinf}.
\begin{table}[!b]
\caption{Axioms for the projection operators and AIP} 
\label{axioms-BTAinf}
\begin{eqntbl}
\begin{axcol}
\proj{0}(x) = \DeadEnd                                  & \axiom{PR1} \\
\proj{n+1}(\DeadEnd) = \DeadEnd                         & \axiom{PR2} \\
\proj{n+1}(\Stop) = \Stop                               & \axiom{PR3} \\
\proj{n+1}(\pcc{x}{\alpha}{y}) =
\pcc{\proj{n}(x)}{\alpha}{\proj{n}(y)}                  & \axiom{PR4}
\eqnsep
\LAND{n \geq 0} \proj{n}(x) = \proj{n}(y) \Limpl x = y  & \axiom{AIP}
\end{axcol}
\end{eqntbl}
\end{table}
In this table, $\alpha$ stands for an arbitrary basic action from 
$\BActTau$ and $n$ stands for an arbitrary natural number from $\Nat$.
We write \BTAinf\ for \BTA\ extended with the projection operators, 
the axioms for the projection operators, and AIP.

By AIP, we have to deal in \BTAinf\ with conditional equational formulas 
with a countably infinite number of premises.
Therefore, infinitary conditional equational logic is used in deriving 
equations from the axioms of \BTAinf.
A complete inference system for infinitary conditional equational logic 
can be found in, for example, \cite{GV93}.

For a simple example of the use of the axioms for the projection 
operators and AIP, we consider the (recursion) equations $x = a \bapf x$ 
and $y = a \bapf a \bapf y$.
With these equations as hypotheses, the following equations are 
derivable from the axioms for the projection operators:  
\begin{ldispl}
\begin{array}[t]{@{}l@{\qquad}l}
\proj{0}(x) = \DeadEnd\;, & \proj{0}(y) = \DeadEnd\;,
\\
\proj{1}(x) = a \bapf \DeadEnd\;, & \proj{1}(y) = a \bapf \DeadEnd\;,
\\
\proj{2}(x) = a \bapf a \bapf \DeadEnd\;, & 
\proj{2}(y) = a \bapf a \bapf \DeadEnd\;,
\\
\proj{3}(x) = a \bapf a \bapf a \bapf \DeadEnd\;, &  
\proj{3}(y) = a \bapf a \bapf a \bapf \DeadEnd\;,
\\
\phantom{\proj{3}(y) = a \bapf a \bapf a \bapf \DeadEnd\;,} \quad \vdots 
\end{array}
\end{ldispl}%
Hence, the conditional equation 
$x = a \bapf x \Land y = a \bapf a \bapf y \Limpl x = y$ 
is derivable from the axioms for the projection operators and AIP.
This conditional equation tells us that the recursion equations 
$x = a \bapf x$ and $y = a \bapf a \bapf y$ have the same solution.

\section{Thread Extraction and Behavioural Congruence}
\label{sect-TE-BC}

In this section, we make precise in the setting of \BTAinf\ which 
behaviours are produced by \PGA\ instruction sequences under 
execution and introduce the notion of behavioural congruence on \PGA\ 
instruction sequences.

To make precise which behaviours are produced by \PGA\ instruction 
sequences under execution, we introduce an operator $\extr{\ph}$ meant 
for extracting from each \PGA\ instruction sequence the thread that 
models the behaviour produced by it under execution.
For each closed \PGA\ term $t$, $\extr{t}$ represents the thread that
models the behaviour produced by the instruction sequence represented 
by $t$ under execution.

Formally, we combine \PGA\ with \BTAinf\ and extend the combination with 
the \emph{thread extraction} operator $\funct{\extr{\ph}}{\InSeq}{\Thr}$ 
and the axioms given in Table~\ref{axioms-thread-extr}.%
\begin{table}[!t]
\caption{Axioms for the thread extraction operator} 
\label{axioms-thread-extr}
\begin{eqntbl}
\begin{axcol}
\extr{a} = a \bapf \DeadEnd                            & \axiom{TE1}  \\
\extr{a \conc X} = a \bapf \extr{X}                    & \axiom{TE2}  \\
\extr{\ptst{a}} = a \bapf \DeadEnd                     & \axiom{TE3}  \\
\extr{\ptst{a} \conc X} = \pcc{\extr{X}}{a}{\extr{\fjmp{2} \conc X}}
                                                       & \axiom{TE4}  \\
\extr{\ntst{a}} = a \bapf \DeadEnd                     & \axiom{TE5}  \\
\extr{\ntst{a} \conc X} = \pcc{\extr{\fjmp{2} \conc X}}{a}{\extr{X}}
                                                       & \axiom{TE6}
\end{axcol}
\qquad
\begin{axcol}
\extr{\fjmp{l}} = \DeadEnd                             & \axiom{TE7}  \\
\extr{\fjmp{0} \conc X} = \DeadEnd                     & \axiom{TE8}  \\
\extr{\fjmp{1} \conc X} = \extr{X}                     & \axiom{TE9}  \\
\extr{\fjmp{l+2} \conc u} = \DeadEnd                   & \axiom{TE10} \\
\extr{\fjmp{l+2} \conc u \conc X} = \extr{\fjmp{l+1} \conc X}
                                                       & \axiom{TE11} \\
\extr{\halt} = \Stop                                   & \axiom{TE12} \\
\extr{\halt \conc X} = \Stop                           & \axiom{TE13}
\end{axcol}
\end{eqntbl}
\end{table}
In this table, 
$a$ stands for an arbitrary basic instruction from $\BInstr$, 
$u$ stands for an arbitrary primitive instruction from $\PInstr$, and 
$l$ stands for an arbitrary natural number from $\Nat$.
We write \PGABTA\ for the combination of \PGA\ and \BTAinf\ extended 
with the thread extraction operator and the axioms for the thread 
extraction operator.
The syntax of closed \PGABTA\ terms of sort $\Thr$ can be defined in 
Backus-Naur style as follows: 
\begin{ldispl}
\nm{CT}'_\Thr \Sis 
\DeadEnd \Sor \Stop \Sor 
(\pcc{\nm{CT}'_\Thr}{\alpha}{\nm{CT}'_\Thr}) \Sor
\extr{\nm{CT}_\InSeq}\;, 
\end{ldispl}%
where $\alpha \in \BActTau$.
$\nm{CT}_\InSeq$ is defined in Section~\ref{sect-PGA}.

A simple example of thread extraction is
\begin{ldispl}
\extr{\ptst{a} \conc \fjmp{2} \conc \fjmp{3} \conc b \conc \halt} =
\pcc{(b \bapf \Stop)}{a}{\DeadEnd}\;.
\end{ldispl}%
In the case of infinite instruction sequences, thread extraction yields 
threads definable by means of a set of recursion equations.
For example, 
\begin{ldispl}
\extr{(\ptst{a} \conc \fjmp{2} \conc \fjmp{3} \conc b \conc \halt)\rep}
\end{ldispl}%
is the solution of the set of recursion equations that consists of the 
single equation
\begin{ldispl}
x = \pcc{(b \bapf \Stop)}{a}{x}\;.
\end{ldispl}%

If a closed \PGA\ term $t$ represents an instruction sequence that
starts with an infinite chain of forward jumps, then TE9 and TE11 can 
be applied to $\extr{t}$ infinitely often without ever showing that a 
basic action is performed.
In this case, we have to do with inaction and, being consistent with 
that, $\extr{t} = \DeadEnd$ is derivable from the axioms of \PGA\ and 
TE1--TE13.
By contrast, $\extr{t} = \DeadEnd$ is not derivable from the axioms of 
\PGAisc\ and TE1--TE13.
However, if closed \PGA\ terms $t$ and $t'$ represent instruction 
sequences in which no infinite chains of forward jumps occur, then 
$t = t'$ is derivable from the axioms of \PGA\ only if 
$\extr{t} = \extr{t'}$ is derivable from the axioms of \PGAisc\ and 
TE1--TE13.

If a closed \PGA\ term $t$ represents an infinite instruction 
sequence, then we can extract the approximations of the thread modeling
the behaviour produced by that instruction sequence under execution up 
to every finite depth: for each $n \in \Nat$, there exists a closed 
\BTA\ term $t''$ such that $\proj{n}(\extr{t}) = t''$ is derivable 
from the axioms of \PGA, TE1--TE13, the axioms of \BTA, and PR1--PR4.
If closed \PGA\ terms $t$ and $t'$ represent infinite instruction 
sequences that produce the same behaviour under execution, then this can
be proved using the following instance of AIP:
$\LAND{n \geq 0} \proj{n}(\extr{t}) = \proj{n}(\extr{t'}) \Limpl
 \extr{t} = \extr{t'}$.

\PGA\ instruction sequences are behaviourally equivalent if they 
produce the same behaviour under execution.
Behavioural equivalence is not a congruence.
Instruction sequences are behaviourally congruent if they produce the
same behaviour irrespective of the way they are entered and the way
they are left.

Let $t$ and $t'$ be closed \PGA\ terms.
Then:
\begin{itemize}
\item
$t$ and $t'$ are \emph{behaviourally equivalent}, 
written $t \beqv t'$, if $\extr{t} = \extr{t'}$ is derivable
from the axioms of \PGABTA.
\item
$t$ and $t'$ are \emph{behaviourally congruent}, written 
$t \bcong t'$, if, for each $l,n \in \Nat$,
$\fjmp{l} \conc t \conc \halt^n \beqv \fjmp{l} \conc t' \conc \halt^n$.%
\footnote
{We use the convention that $t \conc {t'}^0$ stands for $t$.}
\end{itemize}

Some simple examples of behavioural equivalence are
\begin{ldispl}
a \conc \fjmp{2} \conc \ptst{b} \conc \halt \beqv
a \conc \fjmp{2} \conc \ptst{c} \conc \halt\;,
\\
(\ptst{a} \conc \fjmp{2} \conc \fjmp{3} \conc b \conc \halt)\rep \beqv
(\ntst{a} \conc \fjmp{3} \conc b \conc \halt)\rep\;.
\end{ldispl}%
We cannot lift these examples to behavioural congruence, i.e. 
\begin{ldispl}
a \conc \fjmp{2} \conc \ptst{b} \conc \halt \not\bcong
a \conc \fjmp{2} \conc \ptst{c} \conc \halt\;,
\\
(\ptst{a} \conc \fjmp{2} \conc \fjmp{3} \conc b \conc \halt)\rep 
 \not\bcong
(\ntst{a} \conc \fjmp{3} \conc b \conc \halt)\rep\;.
\end{ldispl}%
A simple example of behavioural congruence is
\begin{ldispl}
(\ptst{a} \conc \fjmp{3} \conc \fjmp{2} \conc b)\rep \bcong
(\ntst{a} \conc \fjmp{3} \conc \fjmp{2} \conc b)\rep\;.
\end{ldispl}%

It is proved in~\cite{BM12b} that each closed \PGA\ term is 
behaviourally equivalent to a term of the form $t\rep$, where $t$ is a 
closed repetition-free \PGA\ term.
\begin{proposition}
\label{prop-ICF}
For all closed \PGA\ terms $t$, there exists a closed repetition-free 
\PGA\ term $t'$ such that $t \beqv {t'}\rep$.
\end{proposition}
Behavioural congruence is the largest congruence contained in
behavioural equivalence.
Moreover, structural congruence implies behavioural congruence.
\begin{proposition}
\label{prop-scongr-bequiv}
For all closed \PGA\ terms $t$ and $t'$,
$t = t'$ is derivable from the axioms of \PGA\ only if $t \bcong t'$.
\end{proposition}
\begin{proof}
The proof is basically the proof of Proposition~2.2 
from~\cite{BM12b}.
In that proof use is made of the uniqueness of solutions of sets of 
recursion equations where each right-hand side is a \BTA\ term of
the form $\DeadEnd$, $\Stop$ or $\pcc{s}{\alpha}{s'}$ with \BTA\ terms 
$s$ and $s'$ that contain only variables occurring as one of the 
right-hand sides.
This uniqueness follows from AIP (see also Corollary~2.1 
from~\cite{BM12b}). 
\qed
\end{proof} 
Conversely, behavioural congruence does not imply structural
congruence.
For example,
$\ptst{a} \conc \halt \conc \halt \bcong
 \ntst{a} \conc \halt \conc \halt$,
but
$\ptst{a} \conc \halt \conc \halt =
 \ntst{a} \conc \halt \conc \halt$ 
is not derivable from the axioms of \PGA.

In~\cite{BM17a}, we present an equational axiom system for behavioural
congruence that is sound for closed \PGA\ terms and complete for closed 
repetition-free \PGA\ terms.

The following proposition, proved in~\cite{BM12b}, puts the 
expressiveness of \PGA\ in terms of producible behaviours.
\begin{proposition}
\label{prop-expr}
Let $\mathcal{M}$ be a model of \PGABTA.
Then, for each element $p$ from the domain associated with the sort 
$\Thr$ in $\mathcal{M}$, there exists a closed \PGA\ term $t$ such that 
$p$ is the interpretation of $\extr{t}$ in $\mathcal{M}$ iff $p$ is a 
component of the solution of a finite set of recursion equations 
$\set{V = t_V \where V \in \mathcal{V}}$, where $\mathcal{V}$ is a set 
of variables of sort $\Thr$ and each $t_V$ is a \BTA\ term that is not 
a variable and contains only variables from~$\mathcal{V}$.
\end{proposition}
More results on the expressiveness of \PGA\ can be found 
in~\cite{BM12b}.

\section{The Case of Instructions for Boolean Registers}
\label{sect-PGAbr}

In this section, we present the instantiation of \PGA\ in which all 
possible instructions to read out and alter Boolean registers are taken 
as basic instructions.

In this instantiation, it is assumed that a fixed but arbitrary set 
$\Foci$ of \emph{foci} has been given.
Foci serve as names of Boolean registers.

The set of basic instructions used in this instantiation consists of the 
following:
\begin{itemize}
\item
for each $f \in \Foci$ and $\funct{p,q}{\Bool}{\Bool}$,
a \emph{basic Boolean register instruction} $f.\mbr{p}{q}$.
\end{itemize}
We write $\BInstrbr$ for this set.

Each basic Boolean register instruction consists of two parts separated 
by a dot.
The part on the left-hand side of the dot plays the role of the name of 
a Boolean register and the part on the right-hand side of the dot plays 
the role of an operation to be carried out on the named Boolean register 
when the instruction is executed.
The intuition is basically that carrying out the operation concerned 
modifies the content of the named Boolean register and produces as a 
reply a Boolean value that depends on the content of the named Boolean 
register.
More precisely, the execution of a basic Boolean register instruction 
$f.\mbr{p}{q}$ has the following effects:
\begin{itemize}
\item
if the content of the Boolean register named $f$ is $b$ when the 
execution of $f.\mbr{p}{q}$ starts, then its content is $q(b)$ when the 
execution of $f.\mbr{p}{q}$ terminates;
\item
if the content of the Boolean register named $f$ is $b$ when the 
execution of $f.\mbr{p}{q}$ starts, then the reply produced on 
termination of the execution of $f.\mbr{p}{q}$ is $p(b)$.
\end{itemize}
The execution of $f.\mbr{p}{q}$ has no effect on the content of Boolean 
registers other than the one named $f$.

$\Bool \to \Bool$, the set of all unary Boolean functions, consists of 
the following four functions:
\begin{itemize}
\item
the function $\FFunc$, satisfying 
$\FFunc(\False) = \False$ and $\FFunc(\True) = \False$;
\item
the function $\TFunc$, satisfying 
$\TFunc(\False) = \True$ and $\TFunc(\True) = \True$;
\item
the function $\IFunc$, satisfying 
$\IFunc(\False) = \False$ and $\IFunc(\True) = \True$;
\item
the function $\CFunc$, satisfying 
$\CFunc(\False) = \True$ and $\CFunc(\True) = \False$.
\end{itemize}
In~\cite{BM13b,BM13a,BM14a,BM13c,BM18a}, we actually used the operations 
$\mbr{\FFunc}{\FFunc}$, $\mbr{\TFunc}{\TFunc}$, and 
$\mbr{\IFunc}{\IFunc}$, but denoted them by 
$\setbr{0}$, $\setbr{1}$ and $\getbr$, respectively.
In~\cite{BM14e}, we actually used, in addition to these operations, the 
operation $\mbr{\CFunc}{\CFunc}$, but denoted it by $\negbr$.
Two examples of peculiar operations are $\mbr{\FFunc}{\IFunc}$ and 
$\mbr{\TFunc}{\IFunc}$.
Carrying out one of these operations on a Boolean register does not 
modify the content of the Boolean register and produces as a reply, 
irrespective of the content of the Boolean register, always the same 
Boolean value.

We write $\PGABTAbr$ for \PGABTA\ with $\BInstr$ instantiated by 
$\BInstrbr$.
Notice that $\PGABTAbr$ is itself parameterized by a set of foci.
 
In the papers just mentioned, $\Foci$ is instantiated by
\begin{ldispl}
\set{\inbr{i} \where i \in \Natpos} \union
\set{\outbri{i} \where i \in \Natpos} \union
\set{\auxbr{i} \where i \in \Natpos}
\end{ldispl}%
if the computation of functions from $\Bool^n$ to $\Bool^m$ with $m > 1$ 
is in order and
\begin{ldispl}
\set{\inbr{i} \where i \in \Natpos} \union
\set{\outbr} \union
\set{\auxbr{i} \where i \in \Natpos}
\end{ldispl}%
if only the computation of functions from $\Bool^n$ to $\Bool$ is in 
order.
These foci are employed as follows:
\begin{itemize}
\item
the foci of the form $\inbr{i}$ serve as names of Boolean registers that 
are used as input registers in instruction sequences;
\item
the foci of the form $\outbri{i}$ and $\outbr$ serve as names of Boolean 
registers that are used as output registers in instruction sequences;
\item
the foci of the form $\auxbr{i}$ serve as names of Boolean registers 
that are used as auxiliary registers in instruction sequences.
\end{itemize}
The above sets of foci are just examples of sets by which $\Foci$ may be 
instantiated.
In the algebraic theories presented in Sections~\ref{sect-BRFA} 
and~\ref{sect-TSI}, $\Foci$ is not instantiated.

\section{Boolean Register Families}
\label{sect-BRFA}

\PGA\ instruction sequences under execution may interact with the named 
Bool\-ean registers from a family of Boolean registers provided by their 
execution environment.
In this section, we introduce an algebraic theory of Boolean register 
families called \BRFA\ (Boolean Register Family Algebra).
Boolean register families are reminiscent of the Boolean register files 
found in the central processing unit of a computer 
(see e.g.~\cite{SGA89a}).

In \BRFA, as in $\PGABTAbr$, it is assumed that a fixed but arbitrary 
set $\Foci$ of foci has been given.

\BRFA\ has one sort: the sort $\BRegFam$ of 
\emph{Boolean register families}.
To build terms of sort $\BRegFam$, \BRFA\ has the following constants 
and operators:
\begin{itemize}
\item
the
\emph{empty Boolean register family} constant 
$\const{\emptysf}{\BRegFam}$;
\item
for each $f \in \Foci$ and $b \in \Bool \union \set{\Div}$, 
the \emph{singleton Boolean register family} constant
$\const{f.\br{b}}{\BRegFam}$;
\item
the binary \emph{Boolean register family composition} operator
$\funct{\ph \sfcomp \ph}{\BRegFam \x \BRegFam}{\BRegFam}$;
\item
for each $F \subseteq \Foci$, 
the unary \emph{encapsulation} operator 
$\funct{\encap{F}}{\BRegFam}{\BRegFam}$.
\end{itemize}
We assume that there are infinitely many variables of sort $\BRegFam$,
including $u,v,w$.
We use infix notation for the Boolean register family composition 
operator.
Taking this notational convention into account, the syntax of closed 
\BRFA\ terms (of sort $\BRegFam$) can be defined in Backus-Naur style as 
follows:
\begin{ldispl}
\nm{CT}_\BRegFam \Sis 
\emptysf \Sor f.\br{b} \Sor 
(\nm{CT}_\BRegFam \sfcomp \nm{CT}_\BRegFam) \Sor
\encap{F}(\nm{CT}_\BRegFam)\;, 
\end{ldispl}%
where $f \in \Foci$, $b \in \Bool \union \set{\Div}$, and 
$F \subseteq \Foci$.

The Boolean register family denoted by $\emptysf$ is the empty Boolean 
register family.
The Boolean register family denoted by a closed term of the form 
$f.\br{b}$, where $b \in \Bool$, consists of one named Boolean register 
only, the Boolean register concerned is an operative Boolean register 
named $f$ whose content is $b$.
The Boolean register family denoted by a closed term of the form 
$f.\br{\Div}$ consists of one named Boolean register only, the Boolean 
register concerned is an inoperative Boolean register named $f$.
The Boolean register family denoted by a closed term of the form
$t \sfcomp t'$ consists of all named Boolean registers that belong to 
either the Boolean register family denoted by $t$ or the Boolean 
register family denoted by $t'$.
In the case where a named Boolean register from the Boolean register 
family denoted by $t$ and a named Boolean register from the Boolean 
register family denoted by $t'$ have the same name, they collapse to 
an inoperative Boolean register with the name concerned.
The Boolean register family denoted by a closed term of the form 
$\encap{F}(t)$ consists of all named Boolean registers with a name not 
in $F$ that belong to the Boolean register family denoted by $t$.

A simple example of a Boolean register family is
\begin{ldispl}
\auxbr{8}.\br{1} \sfcomp \auxbr{7}.\br{1} \sfcomp 
\auxbr{6}.\br{0} \sfcomp \auxbr{5}.\br{0}  
\\ \quad
{} \sfcomp
\auxbr{4}.\br{1} \sfcomp \auxbr{3}.\br{1} \sfcomp 
\auxbr{2}.\br{1} \sfcomp \auxbr{1}.\br{0}\;.
\end{ldispl}%
This Boolean register family can be seen as a storage cell whose content 
is the bit string $01110011$.
Taking the content of such storage cells for binary representations of
natural numbers, the functions on bit strings of length $8$ that model 
addition, subtraction, and multiplication modulo $2^8$ of natural 
numbers less than~$2^8$ can be computed using the instructions for 
Boolean registers introduced in Section~\ref{sect-PGAbr}.

An inoperative Boolean register can be viewed as a Boolean register 
whose content is unavailable.
Carrying out an operation on an inoperative Boolean register is 
impossible.

The axioms of \BRFA\ are given in Table~\ref{axioms-BRFA}.%
\begin{table}[!b]
\caption{Axioms of \BRFA} 
\label{axioms-BRFA}
{
\begin{eqntbl}
\begin{axcol}
u \sfcomp \emptysf = u                                 & \axiom{BRFC1} \\
u \sfcomp v = v \sfcomp u                              & \axiom{BRFC2} \\
(u \sfcomp v) \sfcomp w = u \sfcomp (v \sfcomp w)      & \axiom{BRFC3} \\
f.\br{b} \sfcomp f.\br{b'} = f.\br{\Div}               & \axiom{BRFC4}
\end{axcol}
\quad
\begin{saxcol}
\encap{F}(\emptysf) = \emptysf                       & & \axiom{BRFE1} \\
\encap{F}(f.\br{b}) = \emptysf       & \mif f \in F    & \axiom{BRFE2} \\
\encap{F}(f.\br{b}) = f.\br{b}       & \mif f \notin F & \axiom{BRFE3} \\
\multicolumn{2}{@{}l@{\quad}}
 {\encap{F}(u \sfcomp v) =
  \encap{F}(u) \sfcomp \encap{F}(v)}                   & \axiom{BRFE4}
\end{saxcol}
\end{eqntbl}
}
\end{table}
In this table, $f$ stands for an arbitrary focus from $\Foci$,
$F$ stands for an arbitrary subset of $\Foci$, and 
$b$ and $b'$ stand for arbitrary values from $\Bool \union \set{\Div}$.
These axioms simply formalize the informal explanation given
above.

The following two propositions, proved in~\cite{BM12b}, concern an 
elimination result and a representation result for closed \BRFA\ terms.
\begin{proposition}
\label{prop-elim-encap}
For all closed \BRFA\ terms $t$, there exists a closed \BRFA\ term $t'$ 
in which encapsulation operators do not occur such that $t = t'$ is 
derivable from the axioms of \BRFA.
\end{proposition}
\begin{proposition}
\label{prop-represent}
For all closed \BRFA\ terms $t$, for all $f \in \Foci$, either 
$t = \encap{\set{f}}(t)$ is derivable from the axioms of \BRFA\ or there 
exists a $b \in \Bool \union \set{\Div}$ such that 
$t = f.\br{b} \sfcomp \encap{\set{f}}(t)$ is derivable fron the axioms 
of \BRFA.
\end{proposition}

In Section~\ref{sect-comput-bool-fnc}, we will use the notation 
$\Sfcomp{i = 1}{n} t_i$.
For each $i \in \Natpos$, let $t_i$ be a terms of sort $\BRegFam$.
Then, for each $n \in \Natpos$, the term $\Sfcomp{i = 1}{n} t_i$ is 
defined by induction on $n$ as follows: $\Sfcomp{i = 1}{1} t_i = t_1$ 
and $\Sfcomp{i = 1}{n+1} t_i = \Sfcomp{i = 1}{n} t_i \sfcomp t_{n+1}$.

\section{Interaction of Threads with Boolean Registers}
\label{sect-TSI}

A \PGA\ instruction sequence under execution may interact with the named 
Bool\-ean registers from the family of Boolean registers provided by its 
execution environment.
In line with this kind of interaction, a thread may perform a basic 
action basically for the purpose of modifying the content of a named 
Boolean register or receiving a reply value that depends on the content 
of a named Boolean register.
In this section, we introduce related operators.

We combine $\PGABTA(\BInstrbr)$ with \BRFA\ and extend the combination 
with the following operators for interaction of threads with Boolean
registers:
\begin{itemize}
\item
the binary \emph{use} operator
$\funct{\ph \sfuse \ph}{\Thr \x \BRegFam}{\Thr}$;
\item
the binary \emph{apply} operator
$\funct{\ph \sfapply \ph}{\Thr \x \BRegFam}{\BRegFam}$;
\item
the unary \emph{abstraction} operator 
$\funct{\abstr{\Tau}}{\Thr}{\Thr}$;
\end{itemize}
and the axioms given in Tables~\ref{axioms-use-apply}.%
\footnote
{We write $t[t'/x]$ for the result of substituting term $t'$ for 
variable $x$ in term $t$.}
\begin{table}[!t]
\caption{Axioms for the use, apply and abstraction operator} 
\label{axioms-use-apply}
\begin{eqntbl}
\begin{saxcol}
\Stop  \sfuse u = \Stop                                  & \axiom{U1} \\
\DeadEnd \sfuse u = \DeadEnd                             & \axiom{U2} \\
(\Tau \bapf x) \sfuse u = \Tau \bapf (x \sfuse u)        & \axiom{U3} \\
(\pcc{x}{f.\mbr{p}{q}}{y}) \sfuse \encap{\set{f}}(u) =
\pcc{(x \sfuse \encap{\set{f}}(u))}
 {f.\mbr{p}{q}}{(y \sfuse \encap{\set{f}}(u))}           & \axiom{U4} \\
(\pcc{x}{f.\mbr{p}{q}}{y}) \sfuse 
(f.\br{b} \sfcomp \encap{\set{f}}(u)) =
\Tau \bapf (x \sfuse (f.\br{q(b)} \sfcomp \encap{\set{f}}(u)))
                            & \mif p(b) = \True  & \axiom{U5} \\
(\pcc{x}{f.\mbr{p}{q}}{y}) \sfuse 
(f.\br{b} \sfcomp \encap{\set{f}}(u)) =
\Tau \bapf (y \sfuse (f.\br{q(b)} \sfcomp \encap{\set{f}}(u)))
                            & \mif p(b) = \False & \axiom{U6} \\
(\pcc{x}{f.\mbr{p}{q}}{y}) \sfuse 
(f.\br{\Div} \sfcomp \encap{\set{f}}(u)) = \DeadEnd                                          
                                                         & \axiom{U7} \\
\proj{n}(x \sfuse u) = \proj{n}(x) \sfuse u              & \axiom{U8} 
\eqnsep
\Stop  \sfapply u = u                                    & \axiom{A1} \\
\DeadEnd \sfapply u = \emptysf                           & \axiom{A2} \\
(\Tau \bapf x) \sfapply u = \Tau \bapf (x \sfapply u)    & \axiom{A3} \\
(\pcc{x}{f.\mbr{p}{q}}{y}) \sfapply \encap{\set{f}}(u) = \emptysf
                                                         & \axiom{A4} \\
(\pcc{x}{f.\mbr{p}{q}}{y}) \sfapply 
(f.\br{b} \sfcomp \encap{\set{f}}(u)) =
x \sfapply (f.\br{q(b)} \sfcomp \encap{\set{f}}(u))
                            & \mif p(b) = \True  & \axiom{A5} \\
(\pcc{x}{f.\mbr{p}{q}}{y}) \sfapply 
(f.\br{b} \sfcomp \encap{\set{f}}(u)) =
y \sfapply (f.\br{q(b)} \sfcomp \encap{\set{f}}(u))
                            & \mif p(b) = \False & \axiom{A6} \\
(\pcc{x}{f.\mbr{p}{q}}{y}) \sfapply 
(f.\br{\Div} \sfcomp \encap{\set{f}}(u)) = \emptysf      & \axiom{A7} \\
\LAND{k \geq n} t[\proj{k}(x)/z] = s[\proj{k}(y)/z] \Limpl 
t[x/z] = s[y/z]                                          & \axiom{A8}
\eqnsep
\abstr{\Tau}(\Stop) = \Stop                             & \axiom{C1} \\
\abstr{\Tau}(\DeadEnd) = \DeadEnd                       & \axiom{C2} \\
\abstr{\Tau}(\Tau \bapf x) = \abstr{\Tau}(x)            & \axiom{C3} \\
\abstr{\Tau}(\pcc{x}{f.\mbr{p}{q}}{y}) =
\pcc{\abstr{\Tau}(x)}{f.\mbr{p}{q}}{\abstr{\Tau}(y)}    & \axiom{C4} \\
\LAND{n \geq 0}{}
 \abstr{\Tau}(\proj{n}(x)) = \abstr{\Tau}(\proj{n}(y)) \Limpl
\abstr{\Tau}(x) = \abstr{\Tau}(y)                       & \axiom{C5}
\end{saxcol}
\end{eqntbl}
\end{table}
In these tables, $f$ stands for an arbitrary focus from $\Foci$, $p$ and 
$q$ stand for arbitrary Boolean functions from $\Bool \to \Bool$,  
$b$ stands for an arbitrary Boolean value from $\Bool$, 
$n$ stands for an arbitrary natural number from $\Nat$, and
$t$ and $s$ stand for arbitrary terms of sort $\BRegFam$.
We use infix notation for the use and apply operators.
We write \PGABTABRI\ for the combination of \PGABTAbr\ and \BRFA\ 
extended with the use operator, the apply operator, the abstraction 
operator, and the axioms for these operators.
The syntax of closed \PGABTABRI\ terms of sort $\Thr$ and $\BRegFam$ can 
be defined in Backus-Naur style as follows: 
\begin{ldispl}
\nm{CT}''_\Thr \Sis 
\DeadEnd \Sor \Stop \Sor 
(\pcc{\nm{CT}''_\Thr}{\alpha}{\nm{CT}''_\Thr}) \Sor
\extr{\nm{CT}_\InSeq} 
\\ \phantom{\nm{CT}''_\Thr}\quad {} \Sor
(\nm{CT}''_\Thr \sfuse \nm{CT}'_\BRegFam) \Sor
\abstr{\Tau}(\nm{CT}''_\Thr) \;,
\eqnsep
\nm{CT}'_\BRegFam \Sis 
\emptysf \Sor f.\br{b} \Sor 
(\nm{CT}'_\BRegFam \sfcomp \nm{CT}'_\BRegFam) \Sor
\encap{F}(\nm{CT}'_\BRegFam) 
\\ \phantom{\nm{CT}''_\BRegFam}\quad {} \Sor 
(\nm{CT}''_\Thr \sfapply \nm{CT}'_\BRegFam)\;,
\end{ldispl}%
where $\alpha \in \BInstrbr \union \set{\Tau}$, 
$f \in \Foci$, $b \in \Bool \union \set{\Div}$, $F \subseteq \Foci$.
$\nm{CT}_\InSeq$ is defined in Section~\ref{sect-PGA}.

Axioms U1--U7 and A1--A7 formalize the informal explanation of the use 
operator and the apply operator given below and in addition stipulate 
what is the result of apply if an unavailable focus is involved~(A4) and 
what is the result of use and apply if an inoperative Boolean register 
is involved (U7 and A7).
Axioms U8 and A8 allow of reasoning about infinite threads, and 
therefore about the behaviour produced by infinite instruction sequences 
under execution, in the context of use and apply, respectively.

On interaction between a thread and a Boolean register, the thread 
affects the Boolean register and the Boolean register affects the 
thread.
The use operator concerns the effects of Boolean registers on threads 
and the apply operator concerns the effects of threads on Boolean 
registers.
The thread denoted by a closed term of the form $t \sfuse t'$ and the
Boolean register family denoted by a closed term of the form
$t \sfapply t'$ are the thread and Boolean register family, 
respectively, that result from carrying out the operation that is part 
of each basic action performed by the thread denoted by $t$ on the 
Boolean register in the Boolean register family denoted by $t'$ with the 
focus that is part of the basic action as its name.
When the operation that is part of a basic action performed by a thread 
is carried out on a Boolean register, the content of the Boolean 
register is modified according to the operation concerned and the thread 
is affected as follows: the basic action turns into the internal action 
$\Tau$ and the two ways to proceed reduce to one on the basis of the 
reply value produced according to the operation concerned.

With the use operator the internal action $\Tau$ is left as a trace of 
each basic action that has led to carrying out an operation on a Boolean 
register.
The abstraction operator serves to abstract fully from such internal 
activity by concealing $\Tau$.
Axioms C1--C4 formalizes the concealment of $\Tau$.
Axiom C5 allows of reasoning about infinite threads in the context of 
abstraction.

A simple example of use and apply is 
\begin{ldispl}
\extr
 {\Conc{i = 1}{4} 
   (\ntst{\auxbr{i}.\mbr{\IFunc}{\IFunc}} \conc \fjmp{3} \conc
    \auxbr{i}.\mbr{\FFunc}{\FFunc} \conc \halt \conc
    \auxbr{i}.\mbr{\TFunc}{\TFunc})}
\\ \quad {} \sfuse
\auxbr{4}.\br{1} \sfcomp \auxbr{3}.\br{1} \sfcomp 
\auxbr{2}.\br{1} \sfcomp \auxbr{1}.\br{0} 
\\ {} =
\Tau \bapf \Tau \bapf \Tau \bapf \Tau \bapf \Stop\;,
\eqnsep
\extr
 {\Conc{i = 1}{4} 
   (\ntst{\auxbr{i}.\mbr{\IFunc}{\IFunc}} \conc \fjmp{3} \conc
    \auxbr{i}.\mbr{\FFunc}{\FFunc} \conc \halt \conc
    \auxbr{i}.\mbr{\TFunc}{\TFunc})}
\\ \quad {} \sfapply
\auxbr{4}.\br{1} \sfcomp \auxbr{3}.\br{1} \sfcomp 
\auxbr{2}.\br{1} \sfcomp \auxbr{1}.\br{0} 
\\ {} =
\auxbr{4}.\br{1} \sfcomp \auxbr{3}.\br{1} \sfcomp 
\auxbr{2}.\br{0} \sfcomp \auxbr{1}.\br{1}\;.
\end{ldispl}%
In this example, the behaviour of the instructions sequence under 
execution affects the Boolean registers from the Boolean register 
family such that it corresponds to decrement by one on the natural 
number represented by the combined content of the Boolean registers. 
The equations show that, if the combined content of the Boolean 
registers represents $14$,  
(a)~the Boolean registers reduces the behaviour of the instruction 
sequence under execution to termination after four internal actions and 
(b)~the behaviour of the instruction sequence under execution modifies 
the combined content of the Boolean registers to the binary 
representation of $13$.

The following two propositions are about elimination results for 
closed \PGABTABRI\ terms. \sloppy
\begin{proposition}
\label{prop-elim-use}
For all closed \PGABTABRI\ terms $t$ of sort $\Thr$ in which all 
subterms of sort $\InSeq$ are repetition-free, there exists a closed 
\PGABTAbr\ term $t'$ of sort $\Thr$ such that $t = t'$ is derivable from 
the axioms of \PGABTABRI.
\end{proposition}
\begin{proof}
It is easy to prove by structural induction that, for all closed 
rep\-etition-free \PGABTAbr\ terms $s$ of sort $\InSeq$, there exists a 
closed \PGABTAbr\ term $s'$ of sort $\Thr$ such that $\extr{s} = s'$ is 
derivable from the axioms of \PGABTAbr.
Therefore, it is sufficient to prove the proposition for all closed 
\PGABTABRI\ terms $t$ of sort $\Thr$ in which no subterms of sort 
$\InSeq$ occur.
This is proved similarly to part~(1) of Theorem~3.1 from~\cite{BM12b}. 
\qed
\end{proof}
\begin{proposition}
\label{prop-elim-apply}
For all closed \PGABTABRI\ terms $t$ of sort $\BRegFam$ in which all 
subterms of sort $\InSeq$ are repetition-free, there exists a closed 
\PGABTAbr\ term $t'$ of sort $\BRegFam$ such that $t = t'$ is derivable 
from the axioms of \PGABTABRI.
\end{proposition}
\begin{proof}
As in the proof of Proposition~\ref{prop-elim-use}, it is sufficient to 
prove the proposition for all closed \PGABTABRI\ terms $t$ of sort 
$\BRegFam$ in which no subterms of sort $\InSeq$ occur.
This is proved similarly to part~(2) of Theorem~3.1 from~\cite{BM12b}. 
\qed
\end{proof}

\section{Computing Partial Functions from $\Bool^n$ to $\Bool^m$}
\label{sect-comput-bool-fnc}

In this section, we make precise in the setting of the algebraic theory
\PGABTABRI\ what it means that a given instruction sequence computes a 
given partial function from $\Bool^n$ to $\Bool^m$ ($n,m \in \Nat$).

For each $n,m \in \Nat$, we define the following set:
\begin{ldispl}
\Focibr{n}{m} = 
\set{\inbr{i} \where 1 \leq i \leq n} \union 
\set{\auxbr{i} \where i \geq 1} \union 
\set{\outbri{i} \where 1 \leq i \leq m}\;. \;
\end{ldispl}%
We use the instantiation of \PGABTABRI\ in which the set of foci is 
$\UNION{n,m \in \Nat} \Focibr{n}{m}$.
We write $\Focibrall$ for this set and we write \PGABTABRIiao\ for 
\PGABTABRI\ with $\Foci$ instantiated by $\Focibrall$.

Let $n,m \in \Nat$, let $\pfunct{F}{\Bool^n}{\Bool^m}$,%
\footnote
{We write $\pfunct{f}{\Bool^n}{\Bool^m}$ to indicate that $f$ is partial
 function from $\Bool^n$ to $\Bool^m$.}
and let $t$ be a closed repetition-free \PGABTABRIiao\ term of sort 
$\InSeq$ in which only foci from $\Focibr{n}{m}$ occur.
Then $t$ \emph{computes} $F$ if there exists a $k \in \Nat$ such that:
\begin{itemize}
\item
for all $b_1,\ldots,b_n,b'_1,\ldots,b'_m \in \Bool$ with
$F(b_1,\ldots,b_n) = b'_1,\ldots,b'_m$:
\begin{ldispl}
(\extr{t} \sfuse
  ((\Sfcomp{i=1}{n} \inbr{i}.\br{b_i}) \sfcomp
   (\Sfcomp{i=1}{k} \auxbr{i}.\br{\False}))) \sfapply
  (\Sfcomp{i=1}{m} \outbri{i}.\br{\False}) 
\\ \quad {} = 
\Sfcomp{i=1}{m} \outbri{i}.\br{b'_i}\;;
\end{ldispl}%
\item
for all $b_1,\ldots,b_n \in \Bool$ with
$f(b_1,\ldots,b_n)$ undefined:
\begin{ldispl}
(\extr{t} \sfuse
  ((\Sfcomp{i=1}{n} \inbr{i}.\br{b_i}) \sfcomp
   (\Sfcomp{i=1}{k} \auxbr{i}.\br{\False}))) \sfapply
  (\Sfcomp{i=1}{m} \outbri{i}.\br{\False}) 
\\ \quad {} = 
\emptysf\;.
\end{ldispl}%
\end{itemize}

With this definition, we can establish whether an instruction sequence 
of the kind considered in \PGABTABRIiao\ computes a given partial 
function from $\Bool^n$ to $\Bool^m$ ($n,m \in \Nat$) by equational 
reasoning using the axioms of \PGABTABRIiao.

The following proposition tells us that, for each partial function from 
$\Bool^n$ to $\Bool^m$, there exists an instruction sequence of the kind 
considered here that computes it.
\begin{proposition}
\label{prop-comput-bool-fnc}
For all $n,m \in \Nat$, for all $\pfunct{F}{\Bool^n}{\Bool^m}$, there 
exists a closed repetition-free \PGABTABRIiao\ term $t$ in which only 
basic instructions of the forms $f.\mbr{\FFunc}{\FFunc}$, 
$f.\mbr{\TFunc}{\TFunc}$, and $f.\mbr{\IFunc}{\IFunc}$ with 
$f \in \Focibr{n}{m}$ occur such that $t$ computes $F$.
\end{proposition}
\begin{proof}
As an immediate corollary of the proof of Theorem~5.6 in~\cite{BM12b} we
have the following: for all $n,m \in \Nat$, for all 
$\funct{F}{\Bool^n}{\Bool^m}$, there exists a closed repetition-free 
\PGABTABRIiao\ term $t$ in which only basic instructions of the forms 
$f.\mbr{\FFunc}{\FFunc}$, $f.\mbr{\TFunc}{\TFunc}$, and 
$f.\mbr{\IFunc}{\IFunc}$ with $f \in \Focibr{n}{m}$ occur such that 
$t$ computes $F$.
It is easy to see from the same proof that this corollary generalizes 
from total functions to partial functions.
\qed
\end{proof}

The following proposition tells us that an instruction sequence in which 
not only basic instructions of the forms $f.\mbr{\FFunc}{\FFunc}$, 
$f.\mbr{\TFunc}{\TFunc}$, and $f.\mbr{\IFunc}{\IFunc}$ occur can be 
transformed primitive instruction by primitive instruction to an at most 
linearly longer instruction sequence computing the same function in 
which only basic instructions of the forms $f.\mbr{\FFunc}{\FFunc}$, 
$f.\mbr{\TFunc}{\TFunc}$, and $f.\mbr{\IFunc}{\IFunc}$ occur.

The \emph{functional equivalence} relation $\feqv$ on the set of all 
closed repetition-free \PGABTABRIiao\ terms of sort $\InSeq$ is defined 
by $t \feqv t'$ iff there exist $n,m \in \Nat$ such that:
\begin{itemize}
\item
$t$ and $t'$ are terms in which only foci from $\Focibr{n}{m}$ occur;
\item
there exists a $\pfunct{F}{\Bool^n}{\Bool^m}$ such that $t$ computes $F$
and $t'$ com\-putes $F$.
\end{itemize}

\begin{proposition}
\label{prop-restr-instr-set}
There exists a unary function $\phi$ on the set of all closed 
repetition-free \PGABTABRIiao\ terms of sort $\InSeq$ such that:
\begin{itemize}
\item
$\phi$ is the homomorphic extension of a function $\phi'$ from the set 
of all \PGABTABRIiao\ constants of sort $\InSeq$ to the set of all 
closed repetition-free \PGABTABRIiao\ terms of sort $\InSeq$;
\item
for all closed repetition-free \PGABTABRIiao\ terms $t$ of sort 
$\InSeq$:
\begin{itemize}
\item
$t \feqv \phi(t)$;
\item
$\phi(t)$ is a term in which only basic instructions of the forms
$f.\mbr{\FFunc}{\FFunc}$, $f.\mbr{\TFunc}{\TFunc}$, and 
$f.\mbr{\IFunc}{\IFunc}$ occur;
\item
$\phi(t)$ is at most $3 \cdot p$ primitive instructions longer than $t$, 
where $p$ the number of occurrences of basic instructions in $t$ that 
are not of the form $f.\mbr{\FFunc}{\FFunc}$, $f.\mbr{\TFunc}{\TFunc}$ 
or $f.\mbr{\IFunc}{\IFunc}$. 
\end{itemize}
\end{itemize}
\end{proposition}
\begin{proof}
It follows immediately from part~(3) of Proposition~3.1 in~\cite{BM12b} 
that the definition of $\feqv$ given above is a reformulation of the 
instance of the definition of $\feqv$ given in~\cite{BM15a} where the 
set $\Foci$ of foci is instantiated by $\Focibrall$.
This makes the current proposition a corollary of Proposition~2 and 
Theorem~3 in~\cite{BM15a}.
\qed
\end{proof}

The view put forward in this section on what it means in the setting of 
\PGABTABRI\ that a given instruction sequence computes a given partial 
function from $\Bool^n$ to $\Bool^m$ ($n,m \in \Nat$) is the view taken 
in the work on complexity of computational problems, efficiency of 
algorithms, and algorithmic equivalence of programs presented 
in~\cite{BM13b,BM13a,BM14a,BM14e,BM13c,BM18a}.
We remark that Boolean registers cannot only be used to compute partial 
functions from $\Bool^n$ to $\Bool^m$.
For example, it is shown in~\cite{BM08h} that jump instruction are not
necessary if use can be made of Boolean registers.

\section{Uses for the Theory}
\label{sect-work-done}

In this section, we give a survey of uses for \PGABTABRI.

It is often said that a program is an instruction sequence and, if this 
characterization has any value, it must be the case that it is somehow 
easier to understand the concept of an instruction sequence than to 
understand the concept of a program. 
The first objective of the work on instruction sequences that started 
with~\cite{BL02a}, and of which an enumeration is available 
at~\cite{SiteIS}, is to understand the concept of a program. 
The basis of all this work is the parameterized algebraic theory 
\PGABTA\ extended to deal with the interaction between instruction 
sequences under execution and components of their execution environment.
The body of theory developed through this work is such that its use as a 
conceptual preparation for programming is practically feasible.

The notion of an instruction sequence appears in the work in question as 
a mathematical abstraction for which the rationale is based on the 
objective mentioned above. 
In this capacity, instruction sequences constitute a primary field of 
investigation in programming comparable to propositions in logic and 
rational numbers in arithmetic. 
The structure of the mathematical abstraction at issue has been 
determined in advance with the hope of applying it in diverse 
circumstances where in each case the fit may be less than perfect. 

Until now, the work in question has, among other things, yielded an 
approach to computational complexity where program size is used as 
complexity measure, a contribution to the conceptual analysis of the 
notion of an algorithm, and new insights into such diverse issues as the 
halting problem, program parallelization for the purpose of explicit 
multi-threading and virus detection.

The work done in the setting of \PGABTABRI, which is just an 
instantiation of the above-mentioned basis, includes:
\begin{itemize}
\item
Work yielding an approach to computational complexity in which 
algorithmic problems are viewed as families of functions that consist of 
a function from $\Bool^n$ to $\Bool$ for each natural number n and the 
complexity of such problems is assessed in terms of the length of  
instruction sequences that compute the members of these families. 
Several kinds of non-uniform complexity classes have been introduced.
One kind includes a counterpart of the well-known complexity class 
P/poly and another kind includes a counterpart of the well-known 
complexity class NP/poly (see~\cite{BM13a}).
\item
Work contributing to the conceptual analysis of the notion of an 
algorithm.
Two equivalence relations on instruction sequences have been defined,
an algorithmic equivalence relation and a computational equivalence 
relation.
The algorithmic equivalence relation captures to a reasonable degree 
the intuitive notion that two instruction sequences express the same 
algorithm.
Any equivalence relation that captures the notion that two instruction 
sequences express the same algorithm to a higher degree must be finer 
than the computational equivalence relation (see~\cite{BM14a}).
\item
Work showing that, in the case of computing the parity function on bit
strings of length $n$, for each natural number $n$, shorter instruction 
sequences are possible with the use of an auxiliary Boolean register 
than without the use of auxiliary Boolean registers.
This result supports, in a setting where programs are instruction 
sequences acting on Boolean registers, a basic intuition behind the 
storage of auxiliary data, namely the intuition that this makes possible 
a reduction of the size of a program (see~\cite{BM14e}).
\item
Work providing mathematically precise alternatives to the natural 
language and pseudo code descriptions of the long multiplication 
algorithm and the Karatsuba multiplication algorithm.
One established result is that the instruction sequence expressing the 
latter algorithm is shorter than the instruction sequence expressing the 
former algorithm only if the length of the bit strings involved is 
greater than 256. 
Another result is that in a setting with backward jump instructions the 
long multiplication algorithm can be expressed by an instruction 
sequence that is shorter than both these instruction sequences if the 
length of the bit strings involved is greater than 2 (see~\cite{BM13c}).
\item
Work showing that the problem of deciding whether an instruction 
sequence computes the function modeling the non-zeroness test on natural 
numbers less than $2^n$ with respect to their binary representation by 
bit strings of length $n$, for natural number $n$, can only be 
efficiently solved under the restriction that the length of the 
instruction sequence is close to the length of the shortest possible 
instruction sequences that compute this function (see~\cite{BM18a}).
\end{itemize}

\section{Concluding Remarks}
\label{sect-concl}

We have presented the theory underlying a considerable part of the work 
done so far in a line of research in which issues relating to a wide 
variety of subjects from computer science are rigorously investigated 
thinking in terms of instruction sequences.
The distinguishing feature of this presentation is that it is less 
involved than previous presentations. 
Sections~\ref{sect-PGA}, \ref{sect-BTA}, and~\ref{sect-TE-BC} concern 
the part of the presented theory that is relevant to all the work done 
so far in the line of research referred~to.

The restriction to instructions that operate on Boolean registers is a 
classical restrictions in computer science.
Other such classical restrictions are the restriction to instructions 
that operate on natural number registers in register machines and the 
restriction to instructions that operate on Turing tapes in Turing 
machines (see e.g.~\cite{HMU07a}).
Adaptation of Sections~\ref{sect-PGAbr}, \ref{sect-BRFA}, 
and~\ref{sect-TSI} to these restrictions is rather straightforward
(cf.~\cite{BM09k}).

Notice that we have fixed in Section~\ref{sect-comput-bool-fnc}, for 
each use of a Boolean register that must be distinguished to make 
precise what it means that a given instruction sequence computes a given 
partial function from $\Bool^n$ to $\Bool^m$ ($n,m \in \Nat$), the focus 
by which the Boolean register for that use is named.
Because of this and the required generality, the possibility that the 
same Boolean register is used as both input register and output register
is excluded.
Exclusion of possibilities like this can be circumvented by abandoning 
the fixed assignment of foci to register uses and defining ``$t$ 
computes $f$'' relative to an assignment of foci to register uses.
This approach complicates matters, but seems indispensable to find
conclusive answers to open questions like ``what are the shortest 
instruction sequences that compute the function on bit strings of length 
$n$ that models addition modulo $2^n$ on natural numbers less than 
$2^n$, for $n \in \Natpos$?''.

The instruction sequences with instructions for Boolean registers
considered in this paper constitute essentially a programming language 
in which all variables are Boolean variables.
Such programming languages are actually used in toolkits for software 
model checking that make use of the abstract interpretation technique 
known as predicate abstraction (see e.g.~\cite{BR00b,BR01a}).

\subsection*{Acknowledgements}
We thank an anonymous referee for carefully reading a preliminary 
version of this paper, for pointing out some slips made in it, and for 
posing questions that have led to improvements of the presentation.

\bibliographystyle{splncs03}
\bibliography{IS}

\end{document}